\documentclass[11pt,twoside]{article}

\usepackage{parskip}

\usepackage[margin=1in]{geometry}

\usepackage{amsmath,amsthm,amsfonts,amssymb}
\usepackage{mathtools}
\usepackage{dsfont}
\usepackage{bm}

\usepackage{graphicx}
\usepackage{xcolor}
\usepackage{tikz}

\usepackage{subcaption}
\usepackage{booktabs}
\usepackage{array}
\usepackage{multirow}

\usepackage{algorithm}
\usepackage{algpseudocode}

\usepackage{natbib}
\setcitestyle{authoryear,open={(},close={)}}

\usepackage[colorlinks,linkcolor=blue!50!black,citecolor=blue!50!black,urlcolor=blue!50!black]{hyperref}

\usepackage[noabbrev]{cleveref}

\usepackage{microtype}

\usepackage{setspace}

\usepackage{nicefrac}
\usepackage{comment}

\newtheorem{result}{Result}

\newtheorem{lemma}{Lemma}

\newtheorem{corollary}{Corollary}

\theoremstyle{definition}

\newcommand{\E}{\mathds{E}}
\renewcommand{\P}{\mathds{P}}
\newcommand{\V}{\mathds{V}}

\newcommand{\bx}{\bm{x}}

\newcommand{\bX}{\bm{X}}

\newcommand{\R}{\mathbb{R}}

\newcommand{\cB}{\mathcal{B}}

\newcommand{\cW}{\mathcal{W}}

\usepackage{enumitem}
\usepackage[colorlinks]{hyperref}
\usepackage[noabbrev]{cleveref}

\newcommand{\indep}{\perp\!\!\!\perp}

\newcommand{\bbeta}{\bm{\beta}}

\allowdisplaybreaks

\newtheoremstyle{remarkstyle}%
  {6pt}      
  {6pt}      
  {\itshape} 
  {}         
  {\bfseries}
  {.}        
  {.5em}     
  {}         
\theoremstyle{remarkstyle}
\newtheorem{remark}{Remark}

\begin{document}
	
	\begin{center}
		
		{\bf{\LARGE{High-Dimensional Variance Estimation for the \\[8pt] Generalized Regression Estimator}}}
		
		\vspace*{0.2in}
		
		{\large{
				\begin{tabular}{ccc}
					Kalil {\sc Bouhadra}$^{(a)}$ and Mehdi {\sc Dagdoug}$^{(a)}$ 
				\end{tabular}
		}}
		
		\medskip
		
		\begin{tabular}{c}
			$^{(a)}$ Department of Mathematics and Statistics, McGill University
		\end{tabular}
		
		\vspace{0.2in}
		
		\begin{abstract}
	In survey sampling, the goal is to estimate finite population parameters such as totals, means, and proportions. At the estimation stage, it is common to have access to auxiliary information in the form of covariates known either in aggregate form or for each population unit. These covariates are often used, through models relating them to the variable of interest, to improve efficiency; this approach is known as model-assisted estimation. Modern applications increasingly involve settings where a large number of covariates are observed, sometimes of the same order as the sample size. While this setting offers greater modeling flexibility, it also creates important challenges for inference. In this article, we study variance estimation for the generalized regression (GREG) estimator in high-dimensional regimes. We derive new theoretical results that characterize the high-dimensional asymptotic bias of commonly used variance estimators, including those based on Taylor linearization. Furthermore, under suitable distributional assumptions on the covariates, we show that a cross-validated variance estimator is naturally asymptotically unbiased.

\vspace{0.5cm}

		\end{abstract}
		
	\end{center}

    \section{Introduction} \label{sec:intro}

Model-assisted methods are widely used in survey sampling to improve the efficiency of point estimators by leveraging auxiliary information. In particular, the generalized regression estimator (GREG) provides a flexible framework for incorporating covariates through linear modeling; see, e.g., \citet{cassel1977foundations, sarndal1984cosmetic} for foundational contributions and \citet{sarndal2003model} for a pedagogical treatment. Beyond point estimation, variance estimation plays a central role in practice, as it allows for the construction of confidence intervals and other measures of uncertainty routinely reported by national statistical offices. 

A variety of methods have been proposed for variance estimation of the GREG estimator, including approaches based on Taylor linearization and its $g$-weighted version \citep{sarndal1989weighted, valliant2002variance}, as well as resampling techniques such as the jackknife \citep{duchesne2000note, berger2005jackknife}. More recent contributions on the topic include, among others, \citet{stefan2023bootstrap} and \citet{stefan2024jackknife}, which further develop bootstrap and jackknife methodologies in this context.

Classically, the properties of the GREG estimator and its associated variance estimators have been studied within a low-dimensional asymptotic framework, in which the number of covariates is assumed to be negligible relative to the sample size \citep{robinson1983asymptotic, kott1990estimating}. More precisely, this framework considers a regime where the number of covariates $p$ is fixed, while the sample size $n$ and population size $N$ tend to infinity. Such an approximation is appropriate to model practical situations when the ratio $p/n$ is small, that is, $p/n \approx 0$.

In many modern applications, however, practitioners are confronted with settings where the number of covariates is no longer negligible compared to the sample size. For instance, when $n = 100$ and $p = 30$, the ratio $\kappa = p/n$, here equal to $\kappa = 0.3$, is non negligible. A large number of covariates may also originates from a nonlinear low-dimensional setup via basis expansions of the original covariates (e.g., adding powers and interactions of the original covariates). This has led to a growing interest in high-dimensional regimes in survey sampling, where the number of covariates $p$ increases with the sample size $n$. In this context, \citet{cardot2017calibration} studied calibration based on principal components, while \citet{ta2020generalized, chauvet2022asymptotic, dagdoug2023model} investigated the high-dimensional properties of the GREG estimator. More recently, \citet{eustache2025high} highlighted important limitations of classical variance estimators in such settings. In particular, they showed that standard procedures based on Taylor linearization tend to underestimate the variance, whereas resampling methods such as the jackknife tend to overestimate it. These biases can be substantial and persist asymptotically when $p/n$ does not vanish.

Roughly speaking, these phenomena can be traced back to the high-dimensional behavior of several key quantities, which deviate markedly from their classical low-dimensional counterparts (e.g., residuals, leverages, and $g$-weights). For instance, the underestimation exhibited by Taylor-based variance estimators can be attributed to an underestimation of the variability of the regression residuals in high dimensions. Similar issues have been documented in classical statistics; see, for example, \citet{el2018can, zhao2022adaptively}. 

These effects are closely related to overfitting, which arises when a model-assisted estimator relies on a highly flexible regression function. In such cases, residuals tend to be artificially small, and when plugged into Taylor-type variance estimators, this leads to a systematic underestimation of the true variance \citep{dagdoug2023model}. To address this issue, \citet{dagdoug2023model} proposed a cross-validated variance estimator, which replaces the in-sample residuals with residuals obtained through cross-validation. This idea initially originated from
\citet{opsomer2005selecting}, who introduced related techniques in the context of hyper-parameter tuning for model-assisted estimators based on local polynomials. Although developed from a different perspective, it is also closely connected to the cross-fitted variance estimator studied in \citet{an2026agnostic}.

Although \cite{eustache2025high} highlighted important issues of the classical variance estimators of the GREG in high dimensions, several aspects are only partially understood. For example, the high-dimensional bias of the Taylor variance estimator depends on the population average of the so-called g-weights, denoted $\overline{G}_N$. However, in a fixed-design setting, the high-dimensional behavior of $\overline{G}_N$ is difficult to characterize and is generally unknown. In addition, these results were established under high-dimensional assumptions on the sample leverage scores, whose behavior also depends on the distribution of the covariates. These assumptions were supported empirically, but their theoretical validity was not studied. Finally, under Bernoulli sampling, the bias formulas of \cite{eustache2025high} can be used to construct debiased variance estimators. It is less clear, however, whether these estimators would continue to perform well beyond the Bernoulli setting. This motivates the search for a variance estimator that is \emph{naturally} asymptotically unbiased and valid in all dimensions, rather than one obtained through adjustments specific to a particular setting.

\paragraph{Contributions.} In this paper, we further investigate the problem of variance estimation for the GREG estimator in high-dimensional regimes. First, we show that, under suitable conditions on the covariates, some of the assumptions used in \citet{eustache2025high} hold in classical settings. This includes characterizing the high-dimensional behavior of the g-weights mean $\overline{G}_N$ and of the sample leverages. Second, we study the high-dimensional asymptotic bias of a cross-validated variance estimator and establish that, under appropriate conditions on the covariates and the sampling design, the estimator is asymptotically unbiased in all regimes, and therefore does not require any bias correction. We also provide closed-form expressions for the asymptotic biases of the Taylor variance estimator and its g-weighted version. The empirical results presented confirm that these expressions describe well empirical phenomena.

\paragraph{Outline.} The remainder of the paper is organized as follows. In Section~\ref{sec:setup}, we introduce the model-assisted framework and formally define the problem of interest. Section~\ref{sec:main} presents the main theoretical results. In Section~\ref{sec:simu}, we illustrate these findings through a simulation study. Finally, Section~\ref{sec:final} concludes with a discussion of the limitations of our approach and directions for future research. All proofs are postponed to the Appendix. 

\section{Basic setup} \label{sec:setup}

\subsection{Model-assisted estimation}
Consider a finite population $U_N := \{ 1, 2, ..., N \}$ of size $N$. The measurements of a survey variable $Y$ are denoted $y_i$ for $i \in U_N$. We aim to estimate the finite population mean $$ \mu := \dfrac{1}{N}\sum_{i \in U_N} y_i.$$ A random sample $S_N$ of size $n_N$ is drawn from a sampling design $\mathcal{P}_N$. The sample $S_N$ is equivalently characterized by the tuple $(I_i)_{i\in U_N}$ of sampling indicators satisfying $I_i:=1$ if $i \in S_N$ and $I_i :=0$, otherwise. The first-order and second-order inclusion probabilities are defined respectively as $$\pi_i := \P (i \in S_N), \quad \pi_{ij} := \P(i,j\in S_N), \qquad i,j \in U_N.$$

We assume that $\pi_i >0$ for all $i \in U_N$, under which the Horvitz-Thompson estimator $\widehat{\mu}_{\pi} $ of $\mu$ defined by $$ \widehat{\mu}_{\pi}:= \dfrac{1}{N}\sum_{i \in S_N}\dfrac{y_i}{\pi_i},$$ is design-unbiased for $\mu$. We write $\E_p$ and $\V_p$ to denote the expectation and variance with respect to the sampling design, respectively. We denote by $\Delta_{ij} := \pi_{ij}-\pi_i\pi_j $ the sampling covariance of elements $i,j \in U_N$. 

We consider a framework where $p_N$ covariates $X_1,X_2, ..., X_{p_N}$ are observed for every population element; we denote by $\bx_i$ the measurement of these covariates for element $i \in U_N$. Without loss of generality, we assume that the intercept is included in the first position, thus $X_0 =1$. Although the overall number of covariates is $p_N + 1$ with the intercept, we sometimes write $p_N$ instead, for simplicity.

Model-assisted estimation is commonly used to leverage the predictive power of covariates for the variable of interest. In the particular case of linear regression, the resulting estimator is the so-called Generalized REGression estimator \cite[GREG]{sarndal2003model} defined by 
     	\begin{equation}\label{eq:greg}
\widehat{\mu}_{greg}:= \dfrac{ 1}{N}\left(\sum_{i \in U_N} \bx_i^\top \widehat{\bbeta}_N + \sum_{i \in S_N}\dfrac{y_i -\bx_i^\top \widehat{\bbeta}_N  }{\pi_i}\right),
 	\end{equation} where, assuming that $\boldsymbol{A}_\Pi :=  \sum_{i \in S_N}\pi_i^{-1}\bx_i \bx_i^\top$ is invertible, the weighted least squares coefficients $ \widehat{\bbeta}_N$ are given by
    $$ \widehat{\bbeta}_N :=\left(\sum_{i \in S_N} \dfrac{\bx_i \bx_i^\top }{\pi_i} \right)^{-1}\sum_{i \in S_N} \dfrac{\bx_i y_i }{\pi_i}.$$
In what follows, we assume that $\boldsymbol{A}_\Pi$ is invertible, an assumption in line with the current literature, see, e.g., \cite{chauvet2022asymptotic}. 

\subsubsection*{High-dimensional asymptotics}

We are interested in studying the behavior of variance estimators of $\V_p(\widehat{\mu}_{greg})$ in situations where the number of covariates $p_N$ is smaller than the sample size $n_N$, but of similar order, that is, $p_N/n_N\not\approx 0. $ To model this situation, we embed it into an appropriate sequence of similar configurations. Specifically, we adapt the framework of \cite{isaki1982survey} to accommodate for high-dimensional limiting cases. We consider a sequence of increasing populations $(U_N)_{N \in \mathbb{N}}$. In each population, a random sample $S_N$ of size $n_N$ is selected using a sampling design $\mathcal{P}_N$. Although we require the populations $(U_N)_{N \in \mathbb{N}}$ to be embedded, the samples $(S_N)_{N \in \mathbb{N}}$ are random and need not be. Based on each sample $S_N$ and $p_N$ covariates, a model-assisted estimator $\widehat{\mu}_{greg}$ is defined. This framework, therefore, allows for a number of covariates $p_N$ that is increasing as $N$ increases. More specifically, with $\kappa_N := p_N/n_N$, we consider cases where $(p_N)_{N \in \mathbb{N}}$ satisfies $$\lim_{N \to \infty}  \dfrac{p_N}{n_N}  = \lim_{N \to \infty}  \kappa_N := \kappa_\star \in [0,1).$$
This framework includes: (i) the \emph{low-dimensional} case, where $\kappa_\star = 0$, where the number of covariates is asymptotically negligible with respect of the sample size $n_N$; (ii) the \emph{high-dimensional} case, where $\kappa_\star > 0$, where the number of covariates grows at the same order as $n_N$. The setup we consider, however, does not include the ultra-high dimensional case where $\kappa_*\geq 1$, which would require a different treatment.

\subsubsection*{A joint framework}
The aim of this article is to analyze the high-dimensional bias of design-based variance estimators. However, deriving meaningful theoretical results in a purely design-based setup is challenging. 
To circumvent this difficulty, we follow the technique of \cite{eustache2025high}, which considers a joint framework in which both the survey variable $Y$ and the sample $S_N$ are treated as random. Specifically, we assume that the measurements of the survey variable $(y_i)_{i\in U_N}$ are independent random variables satisfying the following linear model 
\begin{equation*}
    y_i = \bx_i^\top \bbeta + \epsilon_i, \qquad i \in U_N,
\end{equation*} where $\epsilon_i$ satisfies $\E_m[\epsilon_i]=0$ and $\E_m[\epsilon_i^2]:= \sigma^2$. 
We denote by $\E_m$ and $\V_m$ the expectation and variance with respect the distribution of the survey variable $Y$ (treating the covariates $X$ fixed), respectively.

The decomposition 
\begin{equation}\label{eq:decompo}
    \V_{mp} (\widehat{\mu}_{greg} ) = \E_m\left[\V_p\left(\widehat{\mu}_{greg}  \right) \right] + \V_m \left( \E_p \left[ \widehat{\mu}_{greg}\right]\right),
\end{equation}
motivates variance estimators of the type 
\begin{equation}\label{eq:var_est}
  \widehat{V}_{mp} =  \widehat{V}_{design} + \dfrac{\widehat{\sigma}^2}{N},  
\end{equation}
where $\widehat{V}_{design} $ represents an estimator of $\V_p\left(\widehat{\mu}_{greg}  \right)$ and $\widehat{\sigma}^2$ an estimator of $\sigma^2:= \V_m(y_1)$. Given that $\sigma^2$ can be estimated unbiasedly in all configurations $(n_N, p_N)$ with $p_N <n_N$, we assume without loss of generality that  $\sigma^2$ is known. 

The analysis technique we follow consists of studying the joint bias of estimators with the structure of $\widehat{V}_{mp}$. Our main underlying interest, however, remains the design bias of $\widehat{V}_{design}$. Since studying this bias directly is mathematically challenging, we instead analyze the joint bias of $\widehat{V}_{mp}$, which includes both a design component and a model component, with the understanding that the main contribution to the bias is expected to come from the design component.

\begin{remark}
The rationale for the structure of the variance estimator $ \widehat{V}_{mp} $ in \eqref{eq:var_est} is based on the intuition that $\widehat{V}_{design}$ estimates the term $\E_m\left[\V_p\left(\widehat{\mu}_{greg}\right)\right]$, and that $\widehat{\sigma}^2/N$ provides a reasonable estimator of $\V_m\left(\E_p\left[\widehat{\mu}_{greg}\right]\right)$. However, this implicitly relies on the idea that $\E_p\left[\widehat{\mu}_{greg}\right] = \mu$, so that $\V_m\left(\E_p\left[\widehat{\mu}_{greg}\right]\right) = \sigma^2/N$. Since the GREG estimator is biased, this argument is usually justified through asymptotic approximations: in the traditional low-dimensional asymptotic framework, the design bias of the GREG estimator is asymptotically negligible \citep{robinson1983asymptotic}. However, this result typically holds when the number of covariates is fixed. Therefore, the validity of the estimator structure in \eqref{eq:var_est} rests on the implicit assumption that the design bias of the GREG estimator remains asymptotically negligible, that is, $\E_m [\E_p [\widehat{\mu}_{greg} -\mu]^2] = o(N^{-1})$, even when $\kappa_\star>0$. This assumption is supported by existing empirical work (see, for example, the simulation section of \cite{dagdoug2023model}), but a formal proof is still lacking. This is, however, beyond the scope of the article, and we do not investigate it further. The above bias analysis technique, based on this idea, proved to work very well to model empirical phenomena, as shown by \cite{eustache2025high} and in our simulations, see Section \ref{sec:simu}.
\end{remark}

\subsection{The leave-one-out variance estimator} 
In the literature, many estimators $\widehat{V}_{design} $ of the design variance $\V_p\left(\widehat{\mu}_{greg}  \right)$ have been suggested. In \cite{eustache2025high}, the following estimators were investigated.
\begin{enumerate}
    \item 	The Taylor variance estimator, defined by 
    \begin{equation} \label{eq:tay}
       \widehat{V}_{tay}  = \dfrac{1}{N^2} \sum_{i \in S_N} \sum_{j \in S_N} \dfrac{\Delta_{ij}}{\pi_{ij}} \dfrac{\widehat{\epsilon}_{i }}{\pi_i} \dfrac{\widehat{\epsilon}_{j }}{\pi_j}, 
    \end{equation}
with $\widehat{\epsilon}_{i }:= y_i - \bx_i^\top \widehat{\bbeta}_N$ for $i\in S_N$.
    
\item 	The g-weighted Taylor variance estimator, defined by 
\begin{equation} \label{eq:tayG}
    \widehat{V}_{g}  = \dfrac{1}{N^2} \sum_{i \in S_N} \sum_{j \in S_N} \dfrac{\Delta_{ij}}{\pi_{ij}} \dfrac{g_{i,N}\widehat{\epsilon}_{i }}{\pi_i} \dfrac{g_\ell \widehat{\epsilon}_{j }}{\pi_j},
\end{equation}
where for $i\in S_N$, $g_{i,N} := \mathbf{t}_{\bx}^\top \boldsymbol{A}_{\Pi }^{-1} \bx_i$ with  $\mathbf{t}_{\bx} := \sum_{i \in U_N} \bx_i$.

\item The Generalized Jackknife variance estimator \citep{berger2005jackknife}, which has the following closed-form solution
\begin{equation} \label{eq:jk}
    \widehat{V}_{jk} = \dfrac{1}{N^2}  \sum_{i \in S_N} \sum_{j \in S_N} \dfrac{\Delta_{ij}}{\pi_{ij}}  \dfrac{\left(1 - w_i\right) g_{i,N} \  \widehat{\epsilon}_{i }}{\left(1-h_{ii,N} \right)\pi_i} \dfrac{\left(1 - w_j\right) g_j \ \widehat{\epsilon}_{j }}{(1-h_{jj})\pi_j},
\end{equation}
where  $w_i := \left(N\pi_i\right)^{-1}$ and the weighted leverages are defined by 
\begin{equation}\label{eq:leverage}
    h_{ii,N}  := \bx_i^\top \boldsymbol{A}_{\Pi }^{-1} \pi_i^{-1} \bx_i, \qquad i\in S_N.
\end{equation}
\end{enumerate}

Each of these estimators exhibits significant biases in high dimensions. Indeed, the Taylor variance estimators $\widehat{V}_{tay}$ in \eqref{eq:tay} and its g-weighted version $\widehat{V}_{g}$ in \eqref{eq:tayG} suffer from important negative biases, while the Jackknife variance estimator $\widehat{V}_{jk}$ in \eqref{eq:jk} exhibits a large positive bias \citep{eustache2025high}. These three estimators are algebraically very similar, yet show very different behaviors. Informally, there are three components providing a partial explanation to these phenomena. First, the high-dimensional behavior of the residuals $(\widehat{\epsilon}_i)_{i\in S_N}$. In high dimensions, these residuals tend to be \emph{artificially} underestimated.  Second, the high-dimensional behavior of the leverages differs from their classical, low-dimensional behavior. Specifically, when $p_N = p $ is fixed, then it can be shown under mild conditions that $\max_{i\in S_N} h_{ii,N} \to0$ in probability as $N \to \infty$. This does not hold in high dimension anymore since $$\max_{i\in S_N} h_{ii,N} \geq \dfrac{1}{n_N}\sum_{i\in S_N} h_{ii,N} = \kappa_N\xrightarrow[N \to \infty]{\mathbb{P}} \kappa_\star>0.$$ For example, in the extreme case where $p_N = n_N$, then $\widehat{\epsilon}_i = 0 $ for all $i\in S_N$. This holds even in cases where all covariates $X_1, ..., X_p$ are independent of $Y$ and thus have no "real predictive power". This is intimately linked to the behavior of the leverages. Indeed, for the sake of illustration, consider the unweighted OLS estimator where $\V_m (\widehat{\epsilon}_i) =\sigma^2(1-h_{ii,N})$. Then, unless $h_{ii,N} \to 0$ as $N\to \infty$, which may not hold in high dimensions, then the second moment of $\widehat{\epsilon}_i$ does not match that of $\epsilon_i$ since $\V_m(\epsilon_i) = \sigma^2$; in a low-dimensional setting with $p_N = p$ fixed, their second moment would match, asymptotically. Finally, the high-dimensional behavior of the GREG estimator, in particular its variance, is fundamentally different from that in the low-dimensional case; see, for example, the simulation study of \cite{dagdoug2023model}.

An informal summary reads as follows: naive $\widehat{\epsilon}_i$ leads to an underestimation (Taylor), multiplying them by $g_{i,N}$ (g-weighted) still leads to an underestimation, although less severe, and, if we were to divide by $1-h_{ii,N}$ (Jackknife -- albeit a negligible factor), we would obtain an overestimation. This suggests considering an estimator that corrects the in-sample residuals by the factor $1-h_{ii,N}$, while avoiding the additional g-weighting present in the jackknife estimator. This approach would lead to the following estimator of the design variance: 
    \begin{equation} \label{eq:loo}
       \widehat{V}_{loo}  = \dfrac{1}{N^2} \sum_{i \in S_N} \sum_{j \in S_N} \dfrac{\Delta_{ij}}{\pi_{ij}} \dfrac{\widehat{\epsilon}_{i }}{(1-h_{ii,N})\pi_i} \dfrac{\widehat{\epsilon}_{j }}{(1-h_{jj, N})\pi_j}.
    \end{equation}
Although our discussion motivating this estimator was informal, it is a very natural estimator in this regime and is known in the literature. Specifically, recognizing that the residual $\widehat{\epsilon}_i^{\ (i)} $ of element $i \in S_N$ that would have been obtained if $\widehat{\bbeta}_N$ were fitted without element $i$ satisfies $\widehat{\epsilon}_i^{\ (i)} = \widehat{\epsilon}_i/(1-h_{ii,N})$, we observe that        $\widehat{V}_{loo} $ corresponds the Leave-One-Out (LOO) variance estimator replacing potentially overfitted residuals $(\widehat{\epsilon}_i)_{i\in S_N}$ by the leave-one-out cross-validation residuals $(\widehat{\epsilon}_i^{\ (i)})_{i\in S_N}$. This estimator traces back to \cite{opsomer2005selecting} which was then used for hyper-parameter tuning. This also corresponds to a particular case of the cross-validated variance estimator in \cite{dagdoug2023model} and the crossfitted variance estimator in \cite{an2026agnostic}. The high-dimensional behavior of the LOO one of the primary interests of this article.

    \section{Main results} \label{sec:main}

    \subsection{Regularity conditions and uniform convergence of sample leverages} \label{sec:unif}

The results presented in this section will be proved under the following regularity conditions. 

\begin{enumerate}[label=(A\arabic*),ref=A\arabic*]    \item \label{A1} The sequence of sampling designs $(\mathcal{P}_N)_{N \in \mathbb{N}}$ is a sequence of Bernoulli designs with parameters $(\pi_N)_{N \in \mathbb{N}}$ with $ \lim_{N \to \infty} \pi_N := \pi_\star$ and $\lim_{N \to \infty}N\pi_N = \infty.$
        \item \label{A2}  The leverages $(h_{ii,N})_{i \in S_N, N \in \mathbb{N}}$ satisfy $$ \max_{i \in S_N} \bigg\rvert h_{ii,N}- \kappa_N \bigg\rvert \xrightarrow[N \to \infty]{\P} 0.$$
\end{enumerate}

Assumption (\ref{A1}) is essentially used to make mathematical arguments more tractable. Although restrictive, we believe that similar conclusions to those presented below would hold in closely related sampling designs such as simple random sampling without replacement. This is supported by the simulations presented in Section \ref{sec:simu}. In general, the aim of the article is to improve our understanding of the phenomena studied in settings where interpretable theoretical results are available.

Assumption (\ref{A2}) is a statement about the uniform convergence of the leverages to $\kappa_\star$ and was initially formulated in \cite{eustache2025high}. In a low-dimensional setting where $p_N = p$ is fixed, this result is easily established. In a high-dimensional setting, it is known that, for each $i \in S_N$, $h_{ii,N}= \kappa_N + o_\P(1) $ \citep{el2018can}.  A uniform statement, such as that of Assumption (\ref{A2}), is stronger and allows for a deeper investigation of the asymptotic bias of the LOO variance estimator. When $p_N^{3/2}\log(n_N)/n_N \to 0 $ but $p_N /\log(n_N) \to \infty$ as $N \to \infty$, Lemma 3.2 of \cite{portnoy1987central} shows that Assumption (\ref{A2}) holds in elliptical distributions. However, the above conditions imply $p_N/n_N \to 0$ as $N \to \infty$, which thus does not include the case $\kappa_\star>0$. We investigate this case in the next lemma. To do so, we also consider the following assumption on the distribution of the covariates. 

\begin{enumerate}[label=(C\arabic*),ref=C\arabic*]    \item \label{C1} The covariates include an intercept $X_0$ and $p_N$ covariates $X_1, X_2, ..., X_{p_N}$ that are independent with Gaussian distribution $\mathcal{N}(0,1)$.
\end{enumerate}

Assumption (\ref{C1}) considers the case where each covariate is independent of the others, with a standard normal distribution. Although restrictive, these types of distributional assumptions on the distribution of the covariates are often needed to establish the high-dimensional behavior of various statistics and are common in the high-dimensional literature. Similar approaches were used, for example, in \cite{portnoy1987central} or \cite{jiang2025new}, more recently. We note that Assumption (\ref{C1}) will only be used for some of our results. We conjecture that our results might hold for other distributions, although different proof strategies would be needed. For example, our simulation study includes covariates drawn from a uniform distribution, for which the empirical behavior of the estimators remains consistent with our theoretical findings.

\begin{lemma}\label{lemma:unifCV}
 Assume (\ref{A1}) and (\ref{C1}). Then, (\ref{A2}) holds, that is, $$ \max_{i \in S_N} \bigg\rvert h_{ii,N}- \kappa_N \bigg\rvert \xrightarrow[N \to \infty]{\P} 0.$$
\end{lemma}
\begin{proof}
    See Appendix \ref{proof:lemma:unifCV}.
\end{proof}

Lemma \ref{lemma:unifCV}, which could be of independent interest, shows that uniform convergence of the leverages extends to the case $\kappa_\star>0$, assuming a Gaussian design. In particular, this shows that Assumption (\ref{A2}) holds under appropriate settings. 

    \subsection{Asymptotically unbiased variance estimation}\label{sec:bias}

In the next result, we determine a closed-form expression for the asymptotic bias of the LOO variance estimator. 

\begin{result} \label{res:asympbiasLOO}
 Assume (\ref{A1}) and (\ref{A2}). Then,
 \begin{equation} \label{eq:ABias}
     \frac{\E_m\!\left(\widehat{V}_{loo}\right)}
{\V_m\!\left(\widehat{\mu}_{greg}\right)}
=
\left(
\frac{1 - \pi_\star}{1 - \kappa_\star} + \pi_\star
\right)
\bigg(
\frac{1}{N} \sum_{i \in U_N} g_{i,N}
\bigg)^{-1} + o_\P(1).
 \end{equation}
\end{result}
\begin{proof}
    See Appendix \ref{proof:asympbiasLOO}.
\end{proof}

Although interesting, the expression \eqref{eq:ABias} is difficult to interpret since it depends on the behavior of the population average of the g-weights $$\overline{G}_N := \frac{1}{N} \sum_{i \in U_N} g_{i,N}.$$ This quantity arises naturally in the high-dimensional asymptotic analysis of various variance estimators \citep{eustache2025high}. Our next lemma analyzes its low and high-dimensional behavior.

    \begin{lemma} \label{lemma:Gweights}
Assume (\ref{A1}). Then, the following statements hold. 
        \begin{enumerate}
        \item[(a)] The average g-weights $\overline{G}_N$ satisfies $$\liminf_{N \to \infty} \overline{G}_N \geq 1 +o_\P(1).$$
            \item[(b)] Assume that
\begin{equation}
\bigg\rVert \dfrac{1}{N}\sum_{i\in S_N} \dfrac{\bx_{i}}{\pi_i} - \dfrac{1}{N}\sum_{i\in U_N} \bx_{i} \bigg\rVert_2= \mathcal{O}_\P\left(\sqrt{\dfrac{p_N}{N}}\right),
\end{equation}
and that there exists a constant $c>0$ such that $$\lim_{N \to \infty} \mathbb{P}\!\left( \lambda_{\min}(\boldsymbol{A}_\Pi/N) \ge c\,
\right) = 1.$$
            Then, provided that $\sqrt{\kappa_N} \max_{i\in U_N} ||\bx_i|| = o_\P(1)$, it holds that $$\max_{i\in U_N} \rvert g_{i,N} - 1\rvert  \xrightarrow[N \to \infty]{\P}0.$$  Consequently, $$\overline{G}_N \xrightarrow[N \to \infty]{\P} 1.$$

            \item[(c)] Assume (\ref{C1}). Then, 
            $$\overline{G}_N \xrightarrow[N \to \infty]{\P} \dfrac{1-\pi_\star}{1-\kappa_\star} + \pi_\star.$$
        \end{enumerate}
    \end{lemma}
\begin{proof}
    See Appendix \ref{lproof:lemma:Gweights}.
\end{proof}

\begin{remark}
Part (b) relies on the condition $\sqrt{\kappa_N} \max_{i \in U_N} \|\bx_i\| = o_\P(1)$, which holds in low-dimensional regimes where $\kappa_N \to 0$, but fails when $\kappa_N \to \kappa_\star > 0$. In the Gaussian setting of (b), using Theorem 3.1.1 of \citet{vershynin2025high}, it can be shown that
\[
\max_{1 \le i \le N} \|x_i\|_2
= \sqrt{p_N} + \mathcal{O}_\P\big( \sqrt{ \log N} \big).
\]
so that $\sqrt{\kappa_N} \max_{i \in U_N} \|\bx_i\|$  vanishes only when $p_N/n_N^{1/2} \to 0$, which does not hold when $\kappa_*>0$. Part (c) characterizes the limit of $\overline{G}_N $ in this regime.
\end{remark}


Lemma \ref{lemma:Gweights} shows that the g-weights have different behaviors in low and high dimensions. In general, (a) shows that the (low and high-dimensional) limit of $\overline{G}_N$ is always greater or equal to $1$. In the low-dimensional case, the g-weights converge uniformly to $1$, which explains why the Taylor variance estimator $\widehat{V}_{tay}$ and its g-weighted version $\widehat{V}_{g}$ in \eqref{eq:tay} and \eqref{eq:tayG}, respectively, share the same asymptotic behavior and performances. When $\kappa_\star>0$, on the other hand, the behavior of the average weights $\overline{G}_N$ changes from converging to $1$ to converging to a function of $\kappa_\star,$ often greater than $1$; this is highlighted in  (ii) in the Gaussian case. 

Our next corollary leverages the high-dimensional characterization of $\overline{G}_N$ to derive closed-form formulas for the asymptotic biases of $\widehat{V}_{tay}$,$\widehat{V}_{g}$ and $\widehat{V}_{loo}$. 

\begin{corollary}  \label{cor:asympbiasLOO}
   Assume (\ref{A1}) and (\ref{C1}). Then, the following statements hold.
    
    \begin{enumerate}
        \item[(i)] The estimator $\widehat{V}_{tay}$ satisfies 
\begin{equation}\label{eq:biasTaylor}
    \frac{\E_m\!\left(\widehat{V}_{tay}\right)}
{\V_m\!\left(\widehat{\mu}_{greg}\right)}
\xrightarrow[N \to \infty]{\P} \frac{(1-\kappa_\star)(1-\kappa_\star(1-\pi_\star))}{1 - \pi_\star \kappa_\star}.
\end{equation}
\item[(ii)] The estimator $\widehat{V}_{g}$ satisfies 
\begin{equation}\label{eq:biasG}
    \frac{\E_m\!\left(\widehat{V}_{g}\right)}
{\V_m\!\left(\widehat{\mu}_{greg}\right)}
\xrightarrow[N \to \infty]{\P} \frac{(1-\kappa_\star)(1-\pi_\star\kappa_\star(1-\pi_\star))}{1 - \pi_\star \kappa_\star}.
\end{equation}
        \item[(iii)] The estimator $\widehat{V}_{loo}$ satisfies 
        $$\frac{\E_m\!\left(\widehat{V}_{loo}\right)}
{\V_m\!\left(\widehat{\mu}_{greg}\right)}
\xrightarrow[N \to \infty]{\P} 1.$$
    \end{enumerate}
    \end{corollary}
    
\begin{proof}
See Appendix \ref{proofCoro1}.
\end{proof}

The above corollary, in (iii) shows that, in the high-dimensional Gaussian regime, the LOO variance estimator is asymptotically unbiased. $\overline{G}_N$ precisely cancels the bias term in \eqref{eq:ABias}. This is therefore an attractive feature of this estimator, which is known to work well in similar scenarios where residuals overfitting may happen \citep{opsomer2005selecting, dagdoug2023model}. On the other hand, (i) and (ii) highlight that, in general, when $\kappa_\star>0$, the Taylor variance estimator $\widehat{V}_{tay}$ and its g-weighted version $\widehat{V}_{g}$  are negatively biased. Although their bias ratios may be somewhat difficult to interpret, the following observations can be made. If either $\kappa_\star = 0$ (low-dimensional case) or $\pi_\star =1$ (only a negligible part of the population is non-sampled), then both bias ratios are equal to $1$, meaning that the estimators are asymptotically unbiased in these settings. However, as soon as $\pi_*<1$ and $\kappa_\star>0$, both bias ratios are strictly less than 1, implying that both estimators $\widehat{V}_{tay}$ and $\widehat{V}_{g}$ are asymptotically (negatively) biased, and, under appropriate conditions, inconsistent. 
Moreover, for any $\pi_\star<1,$ both bias ratios are decreasing as $\kappa_\star$ increases, showing that higher dimensional problems lead to more severe variance underestimations for $\widehat{V}_{tay}$ and $\widehat{V}_{g}$.

    \section{A simulation study: empirical and theoretical behaviors} \label{sec:simu}

    In this section, we present the results of a simulation study comparing the finite-sample performance of the variance estimators $\widehat{V}_{\mathrm{tay}}$, $\widehat{V}_{\mathrm{g}}$, and $\widehat{V}_{loo}$ with the asymptotic benchmark developed in this article.

We generated a finite population of size $N = 1000$ with $X_1, \ldots, X_{p_N} \overset{i.i.d.}{\sim} \mathcal{U}([-1,1])$. To study the effect of the dimension, we considered scenarios with a varying number of covariates,
$p_N = \lfloor \kappa_N \cdot \pi_N \cdot N \rfloor,$
where $\kappa_N \in \{0.1, 0.2, \dots, 0.8\}$ and $\pi_N = n_N/N = 0.3$.

For each scenario, we performed a Monte Carlo simulation study with $B = 1000$ repetitions of the following steps.
\begin{enumerate}
    \item[(i)] Generating the variable of interest $Y$ according to
    \[
      y_i = \bm{x}_i^\top \boldsymbol{\beta} + \epsilon_i,
      \qquad i \in \mathcal{U}_N,
    \]
    where $\boldsymbol{\beta} = \mathbf{1}_{p+1}$, $\epsilon_i \overset{\mathrm{i.i.d.}}{\sim} \mathcal{N}(0, 1)$, and $\bm{x}_i \perp\!\!\!\perp \epsilon_i$ for $i \in \mathcal{U}_N$.

    \item[(ii)] Drawing a sample $S_N$ according to either Simple Random Sampling Without Replacement (SRSWOR) or Bernoulli sampling.

    \item[(iii)] Computing the GREG estimator $\widehat{\mu}_{greg}$ and the three variance estimators considered: the leave-one-out estimator $\widehat{V}_{loo}$, the standard Taylor estimator $\widehat{V}_{tay}$, and the g-weighted Taylor estimator $\widehat{V}_{g}$.
\end{enumerate}

All three estimators were evaluated relative to the design-based variance of $\hat{\mu}_{\mathrm{greg}}$. For each estimator $\widehat{V} \in \{\widehat{V}_{loo}, \widehat{V}_{tay}, \widehat{V}_{g}\}$, we computed their Monte-Carlo Relative Bias (RB) defined as
$$  \mathrm{RB}(\widehat{V})
  \;=\;
  100 \times
  \frac{1}{B}
  \sum_{b=1}^{B}
  \frac{\widehat{V}^{(b)} -  V_{MC} }{V_{MC}}
  \quad (\%),$$
where $V_{MC}$ denotes the Monte Carlo approximation of the total variance of $\hat{\mu}_{\mathrm{greg}}$, computed across the $B$ repetitions. A value $\mathrm{RB}(\widehat{V}) = 0$ corresponds to empirical unbiasedness, while negative values indicate underestimation and positive values indicate overestimation, all expressed in percentages.

To assess the quality of the technique of analysis used in the article, and the usefulness of the formulas provided by Corollary~\ref{cor:asympbiasLOO}, we compared the empirical biases of the variance estimators with the theoretical asymptotic biases given in Corollary~\ref{cor:asympbiasLOO}. Figure~\ref{fig:relative_bias} displays $\mathrm{RB}(\widehat{V}_{loo})$, $\mathrm{RB}(\widehat{V}_{tay})$, and $\mathrm{RB}(\widehat{V}_{g})$ as functions of $\kappa_N$.\\

\begin{figure}[h]
    \centering

    \begin{subfigure}[t]{0.48\textwidth}
        \centering
        \includegraphics[width=\linewidth]{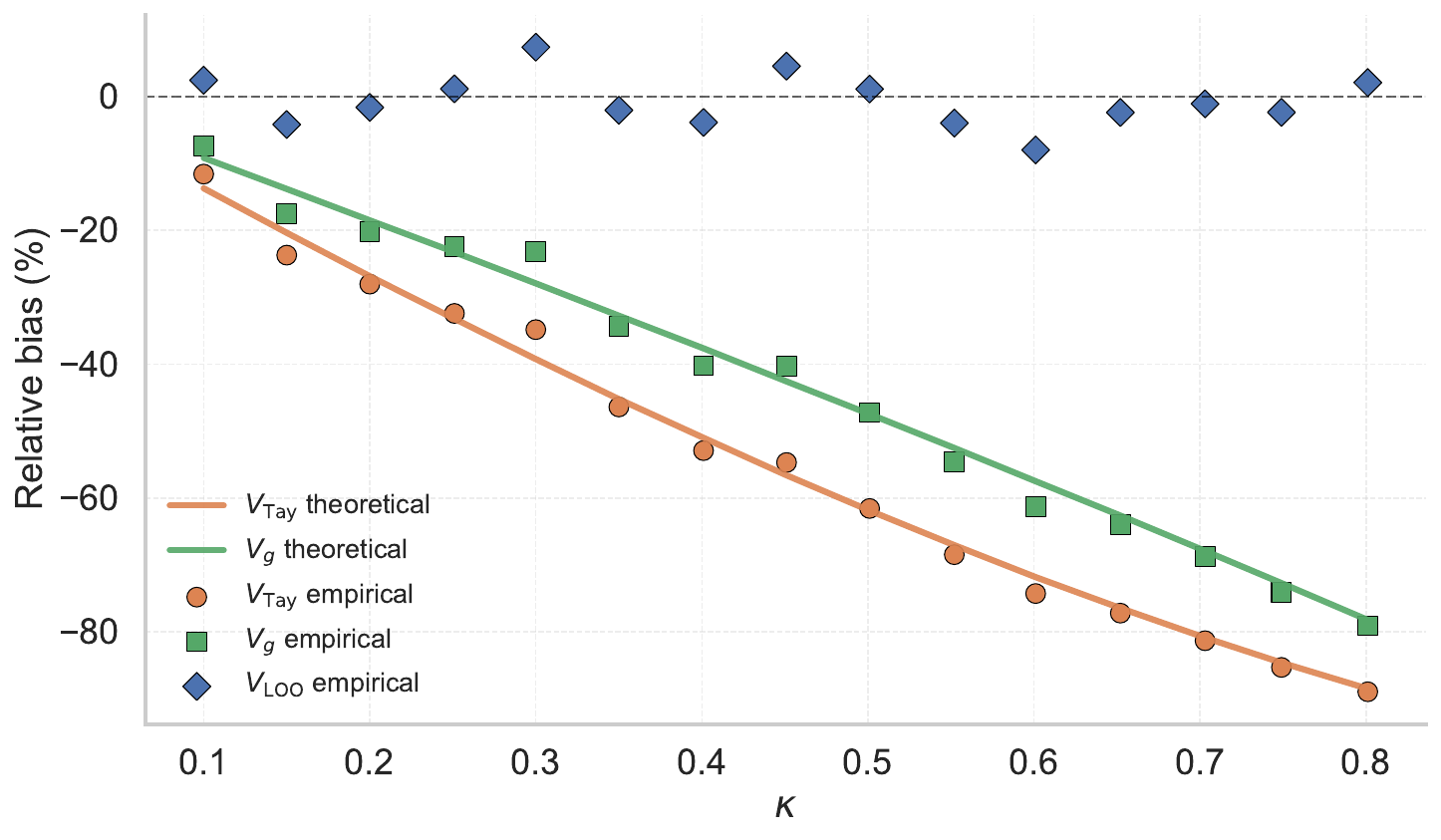}
        \caption{Bernoulli Sampling}
        \label{fig:rb_bern}
    \end{subfigure}
    \hfill
    \begin{subfigure}[t]{0.48\textwidth}
        \centering
        \includegraphics[width=\linewidth]{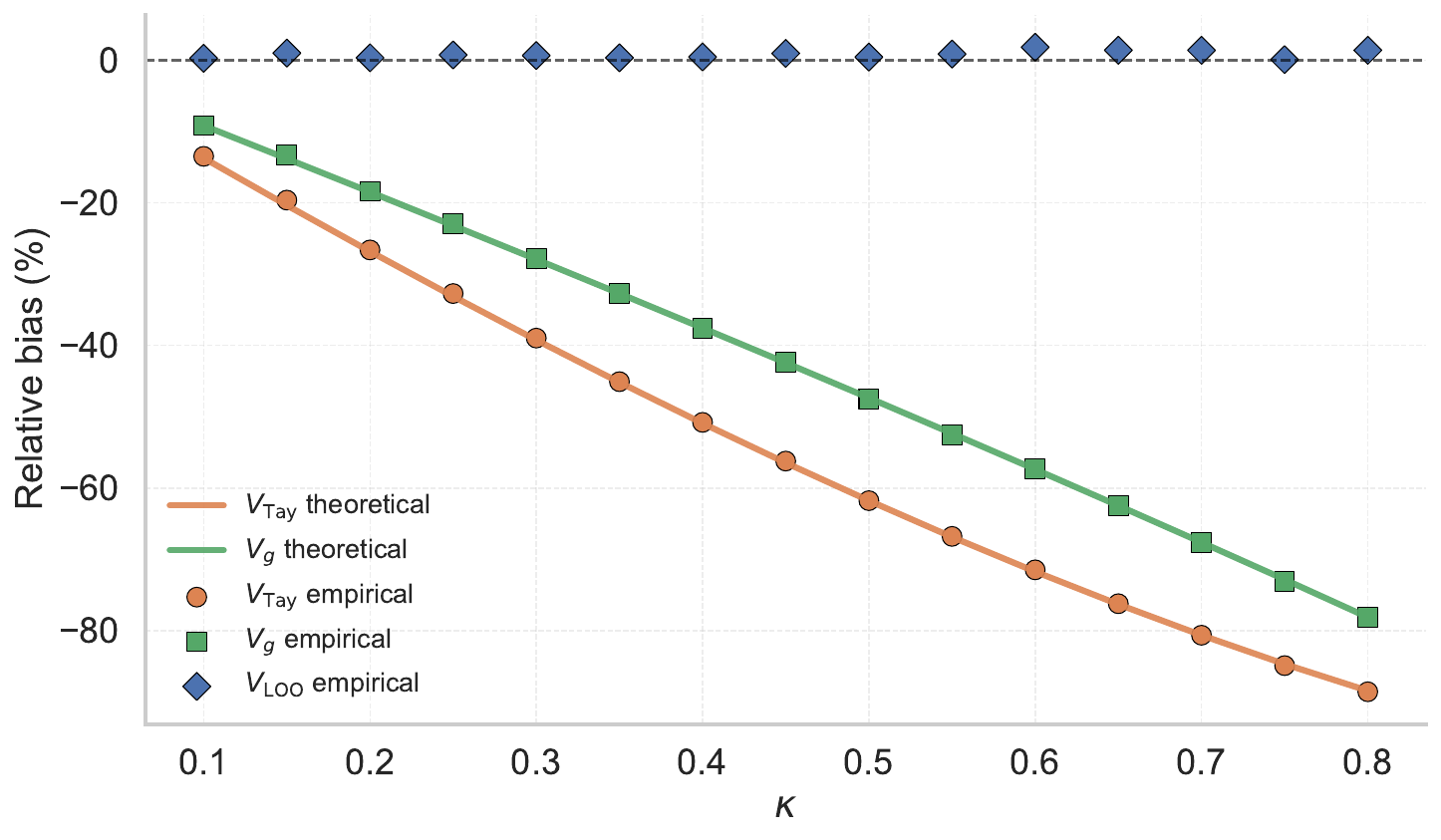}
        \caption{SRSWOR Sampling}
        \label{fig:rb_srswor}
    \end{subfigure}

    \caption{Relative bias (in \%) of the variance estimators $\widehat{V}_{loo}$, $\widehat{V}_{tay}$, and $\widehat{V}_{g}$ as a function of $\kappa_N = p_N / \mathbb{E}(n_N)$ under Bernoulli and SRSWOR sampling designs. The empirical relative biases were represented by the points, while the solid lines correspond to the theoretical curves derived from the asymptotic expressions in \eqref{eq:biasTaylor} and \eqref{eq:biasG}.
    \label{fig:relative_bias}}
\end{figure}

Overall, the results for both sampling designs were in line with the asymptotic results summarized in Corollary~\ref{cor:asympbiasLOO}.
The relative bias of $\widehat{V}_{loo}$ remained close to zero for all values of $\kappa_N$, without requiring any debiasing step depending on the dimension. This seems to confirm that, for SRSWOR and Bernoulli, $\widehat{V}_{loo}$ is a reliable variance estimator of $\widehat{\mu}_{greg}$, independently of the dimension. In contrast, $\widehat{V}_{tay}$ and $\widehat{V}_{g}$ both showed an increasingly negative bias as $\kappa_N$ increased. In particular, the solid curves in Figure~\ref{fig:relative_bias}, which represent the theoretical formulas for the asymptotic biases of $\widehat{V}_{tay}$ and $\widehat{V}_{g}$, agreed well with their empirical behavior. This held for both Bernoulli sampling and SRSWOR sampling.

Finally, we note that similar overall patterns were observed for both sampling designs when using uniform covariates, even though the theory was proved under the Gaussian setting. This suggests that our results may be robust to more general settings.

    \section{Final remarks}

In this article, we studied variance estimation for the GREG estimator in regimes where the number of covariates is not necessarily negligible compared with the sample size. Our results show that the three variance estimators considered, although closely related in form, can have very different high-dimensional behaviors. Under the assumptions of Corollary~\ref{cor:asympbiasLOO}, the leave-one-out estimator $\widehat{V}_{loo}$ is asymptotically unbiased in all dimensional regimes considered. In particular, it does not require any debiasing step depending on the dimension. This is to be contrasted with the standard Taylor estimator $\widehat{V}_{tay}$ and the g-weighted Taylor estimator $\widehat{V}_{g}$, which are negatively biased as soon as $\kappa^\star > 0$ and $\pi^\star < 1$, and this bias becomes more severe as the dimension increases.

The simulation study supports these theoretical findings. It also suggests that the conclusions may be more robust than what is covered by our current proofs. Indeed, similar patterns were observed with uniform covariates, although the main theoretical results were proved under Gaussian covariates. Moreover, the same qualitative behavior was obtained for Bernoulli sampling and SRSWOR. This leads us to believe that the interpretation of the results may extend, at least in spirit, to more general distributional settings and to other sampling designs with equal inclusion probabilities.

Several questions remain open. In particular, the present analysis does not cover unequal inclusion probability designs. In such settings, the behavior of the leverages, the g-weights, and the leave-one-out residuals may be quite different, and it is unclear whether the cancellation mechanism leading to the unbiasedness of $\widehat{V}_{loo}$ would still hold. This is a promising research direction for future work.

    \label{sec:final}
 	\bibliography{biblio}

	\appendix
    \section{Additional notation}
\paragraph{Linear algebra.}
The $i$-th canonical basis vector of $\R^n$ is written $\mathbf{e}_i$. We also denote by $\mathbf{1}_n := [1,...,1]^\top \in \R^n$ the vector of ones in $\R^n$. The span of a vector $\mathbf{u}$ is written $\text{span}(\mathbf{u})$ and the column space of a matrix $\bX$ is denoted $\text{Col}(\bX)$. For a subspace $\mathcal{S} \subset \R^n$, we denote by $\boldsymbol{P}(\mathcal{S})$ the orthogonal projection onto $\mathcal{S}$. In particular, for a matrix $\bX$, we write $\boldsymbol{P}(\bX) := \boldsymbol{P}(\text{Col}(\bX)) = \boldsymbol{X} (\boldsymbol{X}^\top \boldsymbol{X})^{-1} \boldsymbol{X}^\top$. We denote by $s_{\min}(\bX)$ the smallest singular value of a matrix $\bX$, and by $\lambda_{\min}(\boldsymbol{X})$ the smallest eigenvalue of a square matrix $\boldsymbol{X}$. 

\paragraph{Probability distributions.}
We denote by $\overset{d}{=}$ equality in distribution. We write: $Ber(p)$ for the Bernoulli distribution with parameter $p$; $Bin(n,p)$ for the Binomial distribution with parameters $(n,p)$; $\mathcal{N}(\bm{\mu}, \mathbf{\Sigma})$ for the multivariate normal distribution; $\chi^2_k$ for the chi-square distribution with $k$ degrees of freedom; $\mathcal{B}(\alpha,\beta)$ for the Beta distribution with parameters $(\alpha,\beta)$; $\mathcal{W}_p(n,\mathbf{\Sigma})$ for the Wishart distribution with $n$ degrees of freedom and scale matrix $\mathbf{\Sigma}$; $F_{d_1,d_2}$ for the Fisher distribution with $(d_1,d_2)$ degrees of freedom; and $T^2_{p,n}$ for the Hotelling distribution with parameters $(p,n)$.

\paragraph{Specific notation.} The covariates $ \bm{x}_i = (1, X_1, \ldots , X_p)$ decomposes as $ \bm{x}_i = (1, \widetilde{\bm{x}}_i)$. Define the matrix $\boldsymbol{X}_S \in \mathbb{R}^{n \times (p+1)}$ as the matrix whose i-th row is $\bm{x}_i$ and $\widetilde{\boldsymbol{X}}_S \in \mathbb{R}^{n \times p}$ as the matrix whose i-th row is $\widetilde{\bm{x}}_i.$ Similarly, define $\mathbf{t}_{x,S_N} = \sum_{i \in S_N} \bm{x}_i$, $\mathbf{t}_{x,S_N^c} = \sum_{i \in S_N^c} \bm{x}_i$, $\boldsymbol{A}_S = \sum_{i \in S_N}\bx_i \bx_i^\top$ and $\tilde{\mathbf{t}}_{x,S} = \sum_{i \in S_N} \widetilde{\bm{x}}_i$, $\widetilde{\boldsymbol{A}}_S = \sum_{i \in S_N}\widetilde{\bx}_i \widetilde{\bx}_i^\top$.

\begin{remark}
    In the following, although the sample size $n_N$ is random, a law of large numbers argument shows that $n_N/n_{exp,N} \xrightarrow[N \to \infty]{\P}1$, where $n_{exp,N} := N \pi_N$ denotes the expected sample size. For simplicity of notation, we sometimes give asymptotic orders in terms of $n_N$, instead of its deterministic equivalent $n_{exp,N}$. This abuse of notation however does not affect our results. 
\end{remark}

	\section{Proofs of the main results}

	\subsection{Proof of Result \ref{res:asympbiasLOO}} \label{proof:asympbiasLOO}

Recall from \cite{eustache2025high} that
\begin{align*}
\V_m\!\left(\widehat{\mu}_{\mathrm{greg}}\right)
&=
\frac{\sigma^2}{\pi_N N^2}
\sum_{i \in U_N} g_{i,N} .
\end{align*}

Under Bernoulli sampling, the estimator $\widehat{V}_{loo}$ reduces to
\[
\widehat{V}_{loo}
=
\frac{1}{N^2}
\frac{1-\pi_N}{\pi_N^2}
\sum_{i\in S_N}
\frac{\widehat{\epsilon}_{i }^{\ 2}}
{(1-h_{ii,N})^2}
+
\frac{\sigma^2}{N}. 
\]

\noindent
Moreover, $\V_m (\widehat{\epsilon}_i) =\sigma^2(1-h_{ii,N}),$ so that, taking expectation on both sides yields
\begin{align}
\E_m\!\left(\widehat{V}_{loo}\right)
&=
\frac{\sigma^2}{N^2}
\frac{1-\pi_N}{\pi_N^2}
\sum_{i\in S_N}
\frac{1}{1-h_{ii,N}}
+
\frac{\sigma^2}{N}. \label{eq:proof1}
\end{align}

\noindent
 For $i \in S_N$, a first-order Taylor expansion of $(1-h_{ii,N})^{-1}$ around $\kappa_N$ gives
\[
\frac{1}{1 - h_{ii,N}} 
= \frac{1}{1 - \kappa_N} 
+ \dfrac{h_{ii,N} - \kappa_N}{(1-\kappa_N)^2} + \mathcal{O}_\P\bigl(\max_{i \in S_N}|h_{ii,N}-\kappa_N|^2\bigr).
\]
\noindent
Averaging over $i \in S_N$ gives 
\begin{align*}
  \dfrac{1}{n_N}\sum_{i\in S_N}  \frac{1}{1 - h_{ii,N}} 
= \frac{1}{1 - \kappa_N}  
+   \dfrac{1}{n_N}\sum_{i\in S_N}\dfrac{h_{ii,N} - \kappa_N}{(1-\kappa_N)^2} + \mathcal{O}_\P\bigl(\max_{i \in S_N}|h_{ii,N}-\kappa_N|^2\bigr).
\end{align*}
Under \eqref{A2} and using that $\kappa_N \xrightarrow[N \to \infty]{} \kappa_\star$, this reduces to 
\begin{align*}
  \dfrac{1}{n_N}\sum_{i\in S_N}  \frac{1}{1 - h_{ii,N}} 
= \frac{1}{1 - \kappa_\star}  
+ o_\P\left( 1\right).
\end{align*}
Consequently, using \eqref{A1}, equation \eqref{eq:proof1} can be written
\[ 
\E_m\!\left(\widehat{V}_{loo}\right) = \frac{\sigma^2}{\pi_N N} \left(\frac{1-\pi_\star}{1-\kappa_\star} + \pi_\star\right) + o_\P\left(N^{-1}\right).
\]

\noindent
Finally, combining the two previous expressions leads to 
\[
\frac{\E_m\!\left(\widehat{V}_{loo}\right)}
     {\V_m\!\left(\widehat{\mu}_{greg}\right)}
=
\left(
\frac{1-\pi_\star}{1-\kappa_\star} + \pi_\star
\right)
\left(
\frac{1}{N}
\sum_{i \in U_N} g_{i,N}
\right)^{-1} + o_\P(1),
\]
since $(N^{-1}
\sum_{i \in U_N} g_{i,N}
)^{-1} = \mathcal{O}_\P(1)$ by an application of Lemma \ref{lemma:Gweights}.

\subsection{Proof of Corollary \ref{cor:asympbiasLOO} \label{proofCoro1}}

\subsubsection{Proof of (i)}
From Result 4.1 of  \cite{eustache2025high}, we have
\[
\frac{\E_m\!\left(\widehat{V}_{tay}\right)}
     {\V_m\!\left(\widehat{\mu}_{greg}\right)}
=
\frac{(1 - \pi_\star)(1 - \kappa_\star) + \pi_\star}{\overline{G}_N}+o_{\P}(1).
\]
Combining this with Lemma~\ref{lemma:Gweights} (c) yields (i).

\subsubsection{Proof of (ii)}
From Result 4.1 of \cite{eustache2025high}, we have
\begin{equation} \label{eq:Vg_proof}
\frac{\E_m\!\left(\widehat{V}_{g}\right)}
     {\V_m\!\left(\widehat{\mu}_{greg}\right)}
=
(1 - \pi_\star)\left(1- \frac{n_N^{-1}\sum_{i \in S_N}{g_{i,N}^2h_{ii,N}}}{\overline{G}_N}\right)+\frac{\pi_\star}{\overline{G}_N}+o_{\P}(1).
\end{equation}

Let us study the term $n_N^{-1}\sum_{i \in S_N}g_{i,N}^2h_{ii,N}$. 
From (\ref{A2}), we can write
\[
h_{ii,N}=\kappa_N+r_{i,N},
\qquad \text{ where $r_{i,N} := h_{ii,N} - \kappa_N$, } \quad \text{with} \quad 
\max_{i\in S_N}|r_{i,N}|=o_\P(1).
\]
Hence
\[
\dfrac{1}{n_N}\sum_{i\in S_N}g_{i,N}^2h_{ii,N}
=
\kappa_N \dfrac{1}{n_N}\sum_{i\in S_N}g_{i,N}^2
+
\dfrac{1}{n_N}\sum_{i\in S_N}g_{i,N}^2r_{i,N}.
\]
Moreover,
\[
\left|
\dfrac{1}{n_N}\sum_{i\in S_N}g_{i,N}^2r_{i,N}
\right|
\leq
\max_{i\in S_N}|r_{i,N}|
\,
\dfrac{1}{n_N}\sum_{i\in S_N}g_{i,N}^2.
\]

By Technical Lemma 2 of \citet{eustache2025high},
\[
\dfrac{1}{n_N}\sum_{i\in S_N}g_{i,N}^2
=
\frac{\pi_N}{n_N}\sum_{i\in U_N}g_{i,N}
=
\overline G_N.
\]

Since in our setting $ \overline{G}_N = \dfrac{1-\pi_\star}{1-\kappa_\star} + \pi_\star + o_\P(1)$, we have 
\[
n_N^{-1}\sum_{i \in S_N}{g_{i,N}^2h_{ii,N}} = \kappa_N \left( \frac{1-\pi_\star}{1-\kappa_\star} + \pi_\star\right) + o_\P(1).
\]

Finally, replacing $\overline{G}_N$ and $n_N^{-1}\sum_{i \in S_N}{g_{i,N}^2h_{ii,N}}$ in \eqref{eq:Vg_proof}, we obtain
\[
\frac{\E_m\!\left(\widehat{V}_{g}\right)}
     {\V_m\!\left(\widehat{\mu}_{greg}\right)}
=
\frac{(1-\kappa_\star)(1-\pi_\star\kappa_\star(1-\pi_\star))}{1 - \pi_\star \kappa_\star}+o_\P(1).
\]

\subsubsection{Proof of (iii)}
The result follows directly by combining Result~\ref{res:asympbiasLOO} and Lemma~\ref{lemma:Gweights}(c).

\section{Proof of lemmas}

	\subsection{Proof of Lemma \ref{lemma:unifCV}} \label{proof:lemma:unifCV}

Recall that $\boldsymbol{X}_S \in \mathbb{R}^{n \times (p+1)}$ denotes the matrix whose $i$-th row is $\bm{x}_i$, so that $\bm{x}_i = \boldsymbol{X}_S^\top \bm{e}_i$. Therefore, for $i \in S_N$,
\[
h_{ii,N}
=
\bm{x}_i^{\top}
(\boldsymbol{X}_S^{\top}\boldsymbol{X}_S)^{-1}
\bm{x}_i
=
\bm{e}_i^{\top}
\boldsymbol{X}_S
(\boldsymbol{X}_S^{\top}\boldsymbol{X}_S)^{-1}
\boldsymbol{X}_S^{\top}
\bm{e}_i
=
\bm{e}_i^{\top}
\boldsymbol{P}(\boldsymbol{X}_S)
\bm{e}_i.
\]
Since $\boldsymbol{X}_S = [\mathbf{1}_n, \widetilde{\boldsymbol{X}}_S]$, we can apply Lemma \ref{lemma5} with $E = \text{span}(\mathbf{1}_n)$ of dimension $1$, $\boldsymbol{G} = \widetilde{\boldsymbol{X}}_S$ and $\mathbf{a} =\mathbf{e}_i$. Since $\boldsymbol{P}(E)\mathbf{e}_i = \mathbf{1}_n/n_N$ we have $\rVert \boldsymbol{P}(E)\mathbf{e}_i \rVert_2^2 = 1/n_N$. Moreover, by the Pythagorean theorem, $$ \rVert \boldsymbol{P}(E^\perp)\mathbf{e}_i \rVert_2^2 = \rVert \mathbf{e}_i \rVert_2^2 - \rVert \boldsymbol{P}(E)\mathbf{e}_i \rVert_2^2= 1 - \dfrac{1}{n_N}.$$ Therefore, Lemma \ref{lemma5} gives 
\[
h_{ii,N}
\stackrel{d}{=}
\frac{1}{n_N}
+
\left(1 - \frac{1}{n_N}\right) B_i,
\qquad \text{where }
B_i \mid n_N
\sim
\mathcal{B}\!\left(\frac{p_N}{2}, \frac{n_N-1-p_N}{2}\right).
\]
It remains to bound $\max_{i \in S_N} \big\rvert h_{ii,N}- p_N/n_N \big\rvert$ to conclude the proof. Note that 
\begin{align*}
h_{ii,N} - \frac{p_N}{n_N}
&= \frac{1}{n_N} + \left(1-\frac{1}{n_N}\right) B_i - \frac{p_N}{n_N} \\
&= \frac{1}{n_N} + \left(1-\frac{1}{n_N}\right) \left(B_i - \E(B_i \mid n_N)\right), \quad \text{where } \E(B_i \mid n_N) = \frac{p_N}{n_N-1}.
\end{align*}

Therefore,
\[
\max_{i \in S_N}\left|h_{ii,N}-\frac{p_N}{n_N}\right|
\leq
\frac{1}{n_N}
+
\left(1-\frac{1}{n_N}\right)
\max_{i \in S_N}\left|B_i-\E(B_i \mid n_N)\right|.
\]
So it suffices to show that $\max_{i \in S_N}\left|B_i-\E(B_i \mid n_N)\right|=o_\P(1)$. It follows from Theorem 1 of \citet{skorski2023bernstein} that

\[
\mathbb{P}\{|B_i - \E(B_i \mid n_N)| > \epsilon\} \le 2  \exp\!\left(-\dfrac{\epsilon^2}{2\left(v + \frac{|c|\epsilon}
{3}\right)}\right).
\]
where 
\[
v = \frac{p_N(n_N-1-p_N)}{(n_N-1)^2 (n_N+1)} = \mathcal{O}(n_N^{-1}), \qquad
c = \frac{4(n_N-1-2p_N)}{(n_N-1)(n_N+3)} = \mathcal{O}(n_N^{-1}).
\]

Therefore, for every $\epsilon>0$, there exists a constant $C_\epsilon>0$ such that
\[
\P\!\left(\left|B_i - \E(B_i \mid n_N)\right| > \epsilon \mid n_N \right)
\le e^{-C_\epsilon n_N}.
\]

\noindent
Thus, by the union bound over the $n_N$ sampled units,
\[
\P\!\left(\max_{i \in S_N} \left|B_i - \E(B_i \mid n_N)\right| > \epsilon \mid n_N \right)
\le
n_N\,\P\!\left(\left|B_1 - \E(B_1 \mid n_N)\right| > \epsilon \mid n_N\right)
\le
n_N e^{-C_\epsilon n_N}.
\]

\noindent
Taking expectation with respect to the sampling design, and using the fact that $n_N \sim Bin(N,\pi_N)$, we obtain
\begin{align*}
\P\!\left(
\max_{i \in S_N}\bigl|B_i-\E(B_i\mid n_N)\bigr|>\epsilon
\right)
&=
\E\!\left[
\P\!\left(
\max_{i \in S_N}\bigl|B_i-\E(B_i\mid n_N)\bigr|>\epsilon
\,\middle|\, n_N
\right)
\right] \\
&\leq
\E\!\left[n_N e^{-C_\epsilon n_N}\right]  \\
&=  \sum_{k=0}^{N}ke^{-C_\epsilon k}\binom{N}{k}\pi_N^k(1-\pi_N)^{N-k}\\
&= N \pi_N e^{-C\epsilon}\sum_{l=0}^{N-1}\binom{N-1}{l} (\pi_N e^{-C\epsilon})^{l} (1-\pi_N)^{N-1-l} \\
&=
N \pi_N e^{-C_\epsilon} \left(1-\pi_N + \pi_N e^{-C_\epsilon}\right)^{N-1},
\end{align*}
\noindent
which converges to zero as \(N \to \infty\). Hence, $\max_{i \in S_N}\left|B_i - \E(B_i \mid n_N)\right| = o_\P(1)$. This concludes the proof.

	\subsection{Proof of Lemma \ref{lemma:Gweights}} \label{lproof:lemma:Gweights}
        \subsubsection{Proof of (a)}
 We study the term $\overline{G}_N = N^{-1}\mathbf{t}_{\bx}^\top  \boldsymbol{A}_\Pi^{-1} \mathbf{t}_{\bx}.$ To that aim, let $\mathbf{e}_0 := [1, 0, ...,0]^\top \in \R^N$. Since the intercept is included in the covariates, the following equalities hold:
\begin{equation*} 
  \mathbf{e}_0^\top \mathbf{t}_{\bx} = N, \qquad      \mathbf{e}_0^\top\boldsymbol{A}_\Pi^{-1} \mathbf{e}_0 = \dfrac{n_N}{\pi_N}.
\end{equation*}
Therefore, by the Cauchy-Schwarz inequality,
\begin{align*}
    \left(\mathbf{e}_0^\top \mathbf{t}_{\bx}\right)^2 = \left\{ \left( \boldsymbol{A}_\Pi^{1/2} \mathbf{e}_0\right)^\top   \boldsymbol{A}_\Pi^{-1/2} \mathbf{t}_{\bx}\right\}^2\leqslant \rVert \boldsymbol{A}_\Pi^{1/2} \mathbf{e}_0 \rVert_2^2  \rVert \boldsymbol{A}_\Pi^{-1/2}\mathbf{t}_{\bx}\rVert_2^2 = \dfrac{n_N}{\pi_N} \mathbf{t}_{\bx}^\top \boldsymbol{A}_\Pi^{-1} \mathbf{t}_{\bx}.
\end{align*}
Since $    \left(\mathbf{e}_0^\top \mathbf{t}_{\bx}\right)^2 = N^2$, we obtain
$$\overline{G}_N \geq \dfrac{N\pi_N}{n_N}=1 +o_\P(1),$$ which shows the result. 

\subsubsection{Proof of (b)}

Recall that from \citet{sarndal2003model}, $g_{i,N} = 1 + (\mathbf{t}_{\bx} - \widehat{\mathbf{t}}_{\bx})^\top \boldsymbol{A}_{\Pi }^{-1} \bx_i$ where $\widehat{\mathbf{t}}_{\bx} = \sum_{j \in S_N} \bx_j/\pi_j$. Therefore,
\begin{align*}
\max_{i \in U_N}\left|g_{i,N}-1\right|
&= \max_{i \in U_N}\left|(\mathbf{t}_{\bx} - \widehat{\mathbf{t}}_{\bx})^\top \boldsymbol{A}_{\Pi }^{-1} \bx_i\right| \\
&\leq \max_{i \in U_N} \|\mathbf{t}_{\bx} - \widehat{\mathbf{t}}_{\bx}\|\ \cdot \|\boldsymbol{A}_{\Pi }^{-1} \bx_i \| \\
&\leq \|\mathbf{t}_{\bx} - \widehat{\mathbf{t}}_{\bx}\| \cdot \|\boldsymbol{A}_{\Pi }^{-1}\| \cdot \max_{i \in U_N} \|\bx_i\|.
\end{align*}

By hypothesis, $\|\mathbf{t}_{\bx} - \widehat{\mathbf{t}}_{\bx}\| = \mathcal{O}_\P(\sqrt{N p_N})$ and, for $N$ large enough,
\[
\|\boldsymbol{A}_\Pi^{-1}\| = \frac{1}{\lambda_{\min}(\boldsymbol{A}_\Pi)} \le \frac{1}{c\,N} = \mathcal{O}(1/N).
\]

Therefore, we obtain $ \|\mathbf{t}_{\bx} - \widehat{\mathbf{t}}_{\bx}\| \cdot \|\boldsymbol{A}_{\Pi}^{-1}\| = \mathcal{O}_\P(\sqrt{\kappa_N})$, which yields 
\[
\max_{i \in U_N} \rvert g_{i,N} - 1\rvert  \xrightarrow[N \to \infty]{\P}0.
\]

Concerning $\overline{G}_N$, it follows from the triangle inequality that
\[
|\overline G_N-1|
\le \frac1N\sum_{i\in U_N}|g_{i,N}-1|
\le \max_{i\in U_N}|g_{i,N}-1| = o_\P(1).
\]
Therefore, $\overline G_N \xrightarrow[N \to \infty]{\P} 1.$

        \subsubsection{Proof of (c)}

Under Bernoulli sampling, we have
\begin{align*}
\overline G_N
&= \frac{1}{N}\sum_{i \in U_N} \mathbf{t}_{\bx}^\top \boldsymbol{A}_{\Pi }^{-1} \bx_i \\
&= \frac{\pi_N}{N} \mathbf{t}_{\bx}^\top \boldsymbol{A}_{S}^{-1} \mathbf{t}_{\bx} \\
&= \frac{\pi_N}{N}
\left(\mathbf{t}_{x,S_N}^{\top} + \mathbf{t}_{x,S_N^c}^{\top}\right)
\bigg(\sum_{i \in S_N} \bm{x}_i \bm{x}_i^{\top}\bigg)^{-1}
\left(\mathbf{t}_{x,S_N} + \mathbf{t}_{x,S_N^c}\right) \\
&= \frac{\pi_N}{N} \mathbf{t}_{x,S_N}^{\top} \boldsymbol{A}_S^{-1} \mathbf{t}_{x,S_N}
+ \frac{2\pi_N}{N} \mathbf{t}_{x,S_N}^{\top} \boldsymbol{A}_S^{-1} \mathbf{t}_{x,S_N^c}
+ \frac{\pi_N}{N} \mathbf{t}_{x,S_N^c}^{\top} \boldsymbol{A}_S^{-1} \mathbf{t}_{x,S_N^c} \\
&=: A_1(S_N) + A_2(S_N,S_N^c) + A_3(S_N^c).
\end{align*}

\noindent
We now control each of these terms separately. The terms $A_1(S_N)$ and $A_2(S_N, S_N^c)$ are relatively straightforward algebraically, whereas $A_3(S_N^c)$ will require a more detailed analysis.

\noindent\textbf{First and second term: $A_1(S_N)$ and $A_2(S_N,S_N^c)$.} Since $\mathbf{e}_1^{\top}\bm{x}_j = 1,$ for all $j \in U_N$, we obtain 
\[
A_1(S_N) = \frac{\pi_N}{N} \sum_{j \in S_N} \bm{x}_j^\top
\left( \sum_{\ell \in S_N} \bm{x}_\ell \bm{x}_\ell^\top \right)^{-1}
\sum_{i \in S_N} \bm{x}_i = \frac{\pi_N}{N} \sum_{i \in S_N} \mathbf{e}_1^\top \bm{x}_i = \pi_\star^2 + o_\P(1).
\]

Similarly, write
\[
A_2(S_N,S_N^c) = 2\frac{\pi_N}{N} \sum_{j \in S_N} \bm{x}_j^\top
\left( \sum_{\ell \in S_N} \bm{x}_\ell \bm{x}_\ell^\top \right)^{-1}
\sum_{i \in S_N^c} \bm{x}_i = 2\frac{\pi_N}{N} \sum_{i \in S_N^c} \mathbf{e}_1^\top \bm{x}_i = 2\pi_\star (1-\pi_\star) + o_\P(1).
\]

\noindent\textbf{Third term: $A_3(S_N^c)$.} Recall that $A_3(S_N^c)= \pi_N N^{-1} \mathbf{t}_{x,S_N^c}^\top \boldsymbol{A}_S^{-1} \mathbf{t}_{x,S_N^c}$ and 
\[
\mathbf{t}_{x,S_N^c} = \sum_{i\in S_N^c}{\bm{x}}_i=\binom{N-n_N}{\sum_{i\in S_N^c}\widetilde{\bm{x}}_i}.
\]
We first study $\mathbf{t}_{x,S_N^c}$. Since the covariates $\{\widetilde{\bm{x}}_i\}_{i \in U_N}$ are independent with distribution $\mathcal{N}(0, \mathbf{I}_p)$, conditionally on $n_N$, we have
\[
{\sum_{i\in S_N^c}\widetilde{\bm{x}}_i} \overset{d}{=}\sqrt{N-n_N}\widetilde{\mathbf{w}}, \quad \text{where } \widetilde{\mathbf{w}} \sim \mathcal{N}_p(0,\mathbf{I}_p).
\]
Thus, defining
\[
\mathbf{w}:=\binom{1}{\widetilde{\mathbf{w}}},
\qquad
\mathbf{a}_N:=\binom{\sqrt{N(1-\pi_\star)}-1}{\mathbf{0}_p},
\]
we obtain
\begin{align*}
\mathbf{t}_{x,S_N^c}
&= \binom{N-n_N}{\sqrt{N-n_N}\widetilde{\mathbf{w}}} \\
&=\sqrt{N-n_N}\bigg(\binom{1}{\widetilde{\mathbf{w}}} + \binom{\sqrt{N-n_N}-1}{\mathbf{0}_p}\bigg) \\
&= \sqrt{N(1-\pi_\star)}\,(\mathbf{w} + \mathbf{a}_N) + o_\P(N).
\end{align*}

Therefore,
\begin{align*}
A_3(S_N^c)
&=
\frac{\pi_\star}{N}\,
\mathbf{t}_{x,S_N^c}^\top
\boldsymbol{A}_S^{-1}
\mathbf{t}_{x,S_N^c} + o_\P(1)\\
&=
\pi_\star(1-\pi_\star)\,
(\mathbf{w} + \mathbf{a}_N)^\top
\boldsymbol{A}_S^{-1}
(\mathbf{w} + \mathbf{a}_N)
+ o_\P(1) \\
&=
\pi_\star(1-\pi_\star)\Bigl(
\mathbf{w}^\top \boldsymbol{A}_S^{-1}\mathbf{w}
+ 2\,\mathbf{a}_N^\top \boldsymbol{A}_S^{-1}\mathbf{w}
+ \mathbf{a}_N^\top \boldsymbol{A}_S^{-1}\mathbf{a}_N
\Bigr)
+ o_\P(1) \\
&= \pi_\star(1-\pi_\star) (B_1 + B_2 + B_3) + o_\P(1).
\end{align*}

Each of these terms will also be studied independently. 

\noindent\textbf{First term: $B_1$.} It follows directly from Lemma~\ref{lem_3} that
\[
B_1 = \frac{\kappa_\star}{1-\kappa_\star} + o_\P(1).
\]

\noindent\textbf{Second term: $B_2$.} Using \eqref{eq:schur_complement}, $B_2$ simplifies as follows
\begin{align*}
B_2 
&= \binom{\sqrt{N(1-\pi_\star)}-1}{\mathbf{0}_p}^\top \begin{pmatrix}
\Delta_N^{-1}
&
-\Delta_N^{-1}\tilde{\mathbf{t}}_{x,S}^{\top}\widetilde{\boldsymbol{A}}_S^{-1}
\\[0.2cm]
-\widetilde{\boldsymbol{A}}_S^{-1}\tilde{\mathbf{t}}_{x,S} \Delta_N^{-1}
&
\widetilde{\boldsymbol{A}}_S^{-1}
+
\widetilde{\boldsymbol{A}}_S^{-1}\tilde{\mathbf{t}}_{x,S} \Delta_N^{-1}\tilde{\mathbf{t}}_{x,S}^{\top}\widetilde{\boldsymbol{A}}_S^{-1}
\end{pmatrix} \binom{1}{\widetilde{\mathbf{w}}} \\
&= \bigl(\sqrt{N(1-\pi_\star)}-1\bigl) \Delta_N^{-1} - \bigl(\sqrt{N(1-\pi_\star)}-1\bigl) \Delta_N^{-1}\tilde{\mathbf{t}}_{x,S}^{\top}\widetilde{\boldsymbol{A}}_S^{-1}\widetilde{\mathbf{w}} \\
&= o_\P(1),
\end{align*}
where the last equality follows since $\Delta_N^{-1}=(\boldsymbol{A}_S^{-1})_{11}=\mathcal{O}_\P(1/n_N)$ by Lemma~\ref{lem_2}, and $\tilde{\mathbf{t}}_{x,S}^\top \boldsymbol{A}_S^{-1} \widetilde{\mathbf{w}} = \mathcal{O}_\P(1)$, as shown in the proof of Lemma~\ref{lem_3}. 

\noindent\textbf{Third term: $B_3$.} By \eqref{eq:schur_complement} and Lemma~\ref{lem_2}, $B_3$ satisfies
\begin{align*}
B_3
&= \binom{\sqrt{N(1-\pi_\star)}-1}{\mathbf{0}_p}^\top \begin{pmatrix}
\Delta_N^{-1}
&
-\Delta_N^{-1}\tilde{\mathbf{t}}_{x,S}^{\top}\widetilde{\boldsymbol{A}}_S^{-1}
\\[0.2cm]
-\widetilde{\boldsymbol{A}}_S^{-1}\tilde{\mathbf{t}}_{x,S} \Delta_N^{-1}
&
\widetilde{\boldsymbol{A}}_S^{-1}
+
\widetilde{\boldsymbol{A}}_S^{-1}\tilde{\mathbf{t}}_{x,S} \Delta_N^{-1}\tilde{\mathbf{t}}_{x,S}^{\top}\widetilde{\boldsymbol{A}}_S^{-1}
\end{pmatrix} \binom{\sqrt{N(1-\pi_\star)}-1}{\mathbf{0}_p} \\
&= N(1-\pi_{\star})\cdot \Delta_N^{-1} + o_\P(1) \\
&= \frac{(1-\pi_\star)}{\pi_N (1-\kappa_\star)} + o_\P(1).
\end{align*}
Hence,
\[
A_3(S_N^c) = \displaystyle \frac{(1-\pi_\star)^2}{1-\kappa_\star} + \pi_\star(1-\pi_\star)\,\frac{\kappa_\star}{1-\kappa_\star} + o_\P(1).
\]

\noindent\textbf{Conclusion.} Summing the three contributions, we obtain
\begin{align*}
\overline{G}_N 
&= \pi_\star^2 + 2\pi_\star(1-\pi_\star)
+ \frac{(1-\pi_\star)^2}{1-\kappa_\star}
+ \pi_\star(1-\pi_\star)\frac{\kappa_\star}{1-\kappa_\star}
+ o_\P(1) \\
&= \frac{1-\pi_\star}{1-\kappa_\star} + \pi_\star + o_\P(1).
\end{align*}

	\section{Auxiliary lemmas}

\begin{lemma}
\label{lem_2}
Assume (\ref{A1}) and (\ref{C1}). Then,
\[
(\boldsymbol{A}_S^{-1})_{11}
=
\frac{1}{n_N(1-\kappa_\star)} + o_\P\!\left(n_N^{-1}\right).
\]
\end{lemma}

\begin{proof} 
We have
\[ 
\boldsymbol{A}_S
=
\begin{pmatrix}
n_N & \tilde{\mathbf{t}}_{x,S}^{\top} \\
\tilde{\mathbf{t}}_{x,S} & \widetilde{\boldsymbol{A}}_S
\end{pmatrix}.
\]
By the Schur complement formula, $(\boldsymbol{A}_S^{-1})_{11}= (n_N - \tilde{\mathbf{t}}_{x,S}^\top\widetilde{\boldsymbol{A}}_S^{-1} \tilde{\mathbf{t}}_{x,S})^{-1}$. Therefore, it suffices to determine the asymptotic behavior of $\tilde{\mathbf{t}}_{x,S}^\top\widetilde{\boldsymbol{A}}_S^{-1} \tilde{\mathbf{t}}_{x,S}$. Since $\tilde{\mathbf{t}}_{x,S} = \widetilde{\boldsymbol{X}}_S^\top \mathbf{1}_n$ and $\widetilde{\boldsymbol{A}}_S=\widetilde{\boldsymbol{X}}_S^\top\widetilde{\boldsymbol{X}}_S$, we have  $\tilde{\mathbf{t}}_{x,S}^\top\widetilde{\boldsymbol{A}}_S^{-1} \tilde{\mathbf{t}}_{x,S} = \mathbf{1}_n^\top \boldsymbol{P}(\widetilde{\boldsymbol{X}}_S) \mathbf{1}_n$. Hence, the problem reduces to characterizing the distribution of $\mathbf{1}_n^\top \boldsymbol{P}(\widetilde{\boldsymbol{X}}_S) \mathbf{1}_n$. 

Applying Lemma \ref{lemma5} with $E = \{0\}$, $\boldsymbol{G} = \widetilde{\boldsymbol{X}}_S$ and $\mathbf{a} = \mathbf{1}_n$, we get $$\mathbf{1}_n^\top \boldsymbol{P}(\widetilde{\boldsymbol{X}}_S) \mathbf{1}_n \rvert n_N \overset{d}{=} n_N \mathcal{B}\left(\frac{p_N}{2}, \frac{n_N-p_N}{2}\right),$$ since $\rVert \boldsymbol{P}(E^\perp )\mathbf{1}_n\rVert^2 = \rVert \mathbf{1}_n\rVert^2 = n_N.$ Using the moments of the Beta distribution, we have
\[
\E\!\left(\mathbf{1}_n^\top \boldsymbol{P}(\widetilde{\boldsymbol{X}}_S)\mathbf{1}_n \mid n_N \right)
=
p_N, \quad \V\!\left(\mathbf{1}_n^\top \boldsymbol{P}(\widetilde{\boldsymbol{X}}_S)\mathbf{1}_n \mid n_N \right)
= \frac{2\kappa_\star(1-\kappa_\star)n_N^2}{n_N+2}
= 2\kappa_\star(1-\kappa_\star)n_N + o_\P(n_N).
\]

Using the law of total expectation, we obtain $\E(\mathbf{1}_n^\top \boldsymbol{P}(\widetilde{\boldsymbol{X}}_S)\mathbf{1}_n ) = p_N$. Similarly, by the law of total variance, 
\begin{align*}
\V(\mathbf{1}_n^\top \boldsymbol{P}(\widetilde{\boldsymbol{X}}_S)\mathbf{1}_n ) 
&= \V\left( \E\!\left(\mathbf{1}_n^\top \boldsymbol{P}(\widetilde{\boldsymbol{X}}_S)\mathbf{1}_n \mid n_N \right) \right) + \E \left( \V\!\left(\mathbf{1}_n^\top \boldsymbol{P}(\widetilde{\boldsymbol{X}}_S)\mathbf{1}_n \mid n_N \right)\right) \\
&= \E \left( 2\kappa_\star(1-\kappa_\star)n_N + o_\P(n_N) \right) \\
&= 2\kappa_\star(1-\kappa_\star)N \pi_\star + o(N), \quad \text{as  } n_N \sim Bin(N,\pi_N).
\end{align*}

Then, by Chebyshev's inequality, it follows that
\[
\tilde{\mathbf{t}}_{x,S}^\top
\widetilde{\boldsymbol{A}}_S^{-1}
\tilde{\mathbf{t}}_{x,S}
=
p_N + \mathcal{O}_\P(n_N^{1/2}), \quad \text{and} \quad
n_N - \tilde{\mathbf{t}}_{x,S}^\top
\widetilde{\boldsymbol{A}}_S^{-1}
\tilde{\mathbf{t}}_{x,S}
=
n_N(1-\kappa_\star) + o_\P(n_N).
\]

\noindent
Therefore,
\[
(\boldsymbol{A}_S^{-1})_{11}
= \frac{1}{n_N(1-\kappa_\star)} + o_\P(n_N^{-1}).
\]

\end{proof}

\begin{lemma}
\label{lem_3}
Assume (\ref{A1}) and (\ref{C1}). Define $\mathbf{z} = [1 \ \widetilde{\mathbf{z}}^\top]^\top$ where $\widetilde{\mathbf{z}} \sim \mathcal{N}(0, \mathbf{I}_p)$ and $\widetilde{\mathbf{z}} \perp\!\!\!\perp \{\bm{x}_i\}_{i \in U_N}$. Then,
\[
\mathbf z^{\top}\boldsymbol{A}_S^{-1}\mathbf z
=
\frac{\kappa_\star}{1-\kappa_\star}
+
o_\P(1).
\]
\end{lemma}

\begin{proof}  Let $ \Delta_N = n_N - \tilde{\mathbf{t}}_{x,S}^{\top} \widetilde{\boldsymbol{A}}_S^{-1} \tilde{\mathbf{t}}_{x,S}$ be the Schur complement of $\widetilde{\boldsymbol{A}}_S$ in $\boldsymbol{A}_S$. By the block inverse formula, 
\begin{equation}
\boldsymbol{A}_S^{-1}
=
\begin{pmatrix}
\Delta_N^{-1}
&
-\Delta_N^{-1}\tilde{\mathbf{t}}_{x,S}^{\top}\widetilde{\boldsymbol{A}}_S^{-1}
\\[0.2cm]
-\widetilde{\boldsymbol{A}}_S^{-1}\tilde{\mathbf{t}}_{x,S} \Delta_N^{-1}
&
\widetilde{\boldsymbol{A}}_S^{-1}
+
\widetilde{\boldsymbol{A}}_S^{-1}\tilde{\mathbf{t}}_{x,S} \Delta_N^{-1}\tilde{\mathbf{t}}_{x,S}^{\top}\widetilde{\boldsymbol{A}}_S^{-1}
\end{pmatrix}.
\label{eq:schur_complement}
\end{equation}
It follows that 
\begin{align*}
\mathbf z^{\top} \boldsymbol{A}_S^{-1} \mathbf z
=
\Delta_N^{-1}
+
\widetilde{\mathbf z}^{\top}\widetilde{\boldsymbol{A}}_S^{-1}\widetilde{\mathbf z}
+
\Delta_N^{-1}
\bigl(
\tilde{\mathbf{t}}_{x,S}^{\top}\widetilde{\boldsymbol{A}}_S^{-1}\widetilde{\mathbf z}
\bigr)^2
-
2 \Delta_N^{-1}\tilde{\mathbf{t}}_{x,S}^{\top}\widetilde{\boldsymbol{A}}_S^{-1}\widetilde{\mathbf z}
:=
T_1 + T_2 + T_3 + T_4 .
\end{align*}
We study each term separately and show that, $T_2 = \kappa_\star/(1-\kappa_\star) + o_\P(1)$ and that $T_1$, $T_3$ and $T_4$ all converge to zero in probability, from which the result will follow.  

\noindent\textbf{First term: $T_1 = \Delta_N^{-1}$.} A direct application of Lemma \ref{lem_2} shows that $\Delta_N^{-1} = o_\P(1).$

\noindent\textbf{Second term: $T_2 = \widetilde{\mathbf z}^{\top}\widetilde{\boldsymbol{A}}_S^{-1}\widetilde{\mathbf z}$.} Recall that, conditional on $S_N$, the rows of $\widetilde{\boldsymbol{X}}_S$ are independent $\mathcal{N}(0,\mathbf{I}_p)$ under (\ref{C1}). Therefore, 
$\widetilde{\boldsymbol{A}}_S = \widetilde{\bX}_S^\top \widetilde{\bX}_S \sim \cW_p(\boldsymbol{I}_p, n_N)$ which implies that
\[ 
n_N\, \mathbf{\widetilde{z}}^{\top} \boldsymbol{\widetilde{A}}_S^{-1} \mathbf{\widetilde{z}} \sim T^2_{p_N,n_N}, \quad  \text{ and } \quad T_2 \sim \dfrac{p_N}{n_N-p_N+1} F_{p_N,n_N-p_N+1}.
\]

\noindent
Consequently,
\[
\E[T_2]
= \frac{p_N}{n_N-p_N+1} \cdot \frac{n_N-p_N+1}{n_N-p_N-1}
= \frac{\kappa_\star}{1-\kappa_\star}+o(1), 
\]
\[
\V(T_2)
= \frac{p_N^2}{(n_N-p_N+1)^2}
   \cdot
   \frac{2(n_N-p_N+1)^2\big(p_N+n_N-p_N+1-2\big)}
        {p_N\,(n_N-p_N-1)^2\,(n_N-p_N-3)}
= \frac{2}{n_N}\,\frac{\kappa_\star}{(1-\kappa_\star)^3}+o(n_N^{-1}).
\]

\noindent
An application of Chebyshev's inequality gives
\[
T_2
=
\frac{\kappa_\star}{1-\kappa_\star}
+
o_\P(1).
\]

\noindent\textbf{Third term: $T_3 = \Delta_N^{-1}
\bigl(
\tilde{\mathbf{t}}_{x,S}^{\top}\widetilde{\boldsymbol{A}}_S^{-1}\widetilde{\mathbf z}
\bigr)^2$.} We now show that $T_3 = o_\P(1)$. From Lemma~\ref{lem_2},  we have
$\Delta_N^{-1} = o_\P(1)$, so it suffices to show that
$T_3' = \tilde{\mathbf{t}}_{x,S}^{\top}\widetilde{\boldsymbol{A}}_S^{-1}\widetilde{\mathbf z} = \mathcal{O}_\P(1)$. Note that $T_3' \mid \{\bm{x}_i\}_{i \in S_N} 
\sim 
\mathcal N(0,
\|\widetilde{\boldsymbol{A}}_S^{-1}\tilde{\mathbf{t}}_{x,S}\|^2).$ We therefore study the quantity $\|\widetilde{\boldsymbol{A}}_S^{-1}\tilde{\mathbf{t}}_{x,S}\|$. First, observe that
\[
\tilde{\mathbf{t}}_{x,S}
=
\sum_{i\in S_N}\widetilde{\boldsymbol{X}}_i
\sim \mathcal{N}(0, n_N \mathbf{I}_p).
\]
Hence, $n_N^{-1}\|\tilde{\mathbf{t}}_{x,S}\|^2 \sim \chi^2_p$, which implies, by Chebyshev's inequality, that $\|\tilde{\mathbf{t}}_{x,S}\| = \mathcal{O}_\P(\sqrt{p_N n_N}) = \mathcal{O}_\P(n_N)$. Moreover, since $\widetilde{\boldsymbol{A}}_S = \widetilde{\boldsymbol{X}}_S^\top \widetilde{\boldsymbol{X}}_S$, we have
\[
\|\widetilde{\boldsymbol{A}}_S^{-1}\|
=
\lambda_{\min}^{-1}(\widetilde{\boldsymbol{A}}_S)
=
s_{\min}^{-2}(\widetilde{\boldsymbol{X}}_S).
\]

Under \eqref{C1}, in the regime $p_N/n_N \to \kappa_\star \in (0,1)$, Exercice 7.13 of \citet{vershynin2025high} gives that, for any $t \ge 0$,
\[
\mathbb{P}\!\left(
s_{\min}(\widetilde{\boldsymbol{X}}_S)
\ge
\sqrt{n_N} - \sqrt{p_N} - t
\right)
\ge
1 - 2\exp(-t^2).
\]

Let $\epsilon>0$ be small enough and take $t = \epsilon \sqrt{n_N}$. Then there exists a constant $c = c(\kappa_N,\epsilon) > 0$ such that
\[
\mathbb{P}\!\left(
s_{\min}(\widetilde{\boldsymbol{X}}_S) \ge c \sqrt{n_N}
\right)
\ge 1 - 2 \exp(-\epsilon^2 n_N).
\]
Consequently,
\[
\mathbb{P}\!\left(
\|\widetilde{\boldsymbol{A}}_S^{-1}\| \le \frac{1}{c^2 n_N} 
\right)
\ge 1 - 2 \exp(-\epsilon^2 n_N)
\xrightarrow[N \to \infty]{} 1,
\]
so that $\|\widetilde{\boldsymbol{A}}_S^{-1}\|
= \mathcal{O}_\P(n_N^{-1})$. Combining both bounds, we obtain
\[
\|\widetilde{\boldsymbol{A}}_S^{-1}\tilde{\mathbf{t}}_{x,S}\|
\le
\|\widetilde{\boldsymbol{A}}_S^{-1}\|\,
\|\tilde{\mathbf{t}}_{x,S}\|
=
\mathcal{O}_\P(1).
\]

Therefore, by Chebyshev's inequality, $T_3' \mid \{\bm{x}_i\}_{i \in S_N} = \mathcal{O}_\P(1)$ from which it can be shown that the unconditional tightness also holds. Hence, $T_3 = o_\P(1).$

\noindent\textbf{Fourth term: $T_4 = - 2 \Delta_N^{-1}\tilde{\mathbf{t}}_{x,S}^{\top}\widetilde{\boldsymbol{A}}_S^{-1}\widetilde{\mathbf z}$.} Since $T_3' = \tilde{\mathbf{t}}_{x,S}^{\top}\widetilde{\boldsymbol{A}}_S^{-1}\widetilde{\mathbf z} = \mathcal{O}_\P(1)$, and by Lemma~\ref{lem_2}, $\Delta_N^{-1} = o_\P(1)$, their product is negligible in probability and thus, $T_4 = o_\P(1)$.

\end{proof}

\begin{lemma}\label{lemma5}
    Let $\boldsymbol{G}\in \R^{n\times p}$ be a random matrix with independent $\mathcal{N}(0,1)$ entries. Let $E \subset \R^n$ be a fixed deterministic subspace of dimension $q$, such that $p+q <n $. Then, for any vector $\mathbf{a}\in \R^n$, $$ \mathbf{a}^\top \boldsymbol{P}\left( E + \text{Col}(G)\right) \mathbf{a} \overset{d}{=}  ||\boldsymbol{P}\left( E\right) \mathbf{a} \rVert^2 + ||\boldsymbol{P}( E^\perp )\mathbf{a} \rVert^2 \cdot \cB \left( \dfrac{p}{2}, \dfrac{n-q-p}{2}\right).$$
\end{lemma}
\begin{proof}
We have $E + \text{Col}(G) = E\oplus \text{Col}\left(\boldsymbol{P}(E^\perp )\boldsymbol{G}\right) $ in orthogonal direct sum, so that $$\boldsymbol{P}\left(E + \text{Col}(G) \right)= \boldsymbol{P}\left(E\right) +\boldsymbol{P}\left(\boldsymbol{P}(E^\perp )\boldsymbol{G} \right).  $$ Moreover, decomposing $\mathbf{a} = \boldsymbol{P}(E)\mathbf{a}+\boldsymbol{P}(E^\perp)\mathbf{a}$ and noting that $\boldsymbol{P}(E^\perp)\mathbf{a} \in E^\perp$, we get $$ \mathbf{a}^\top \boldsymbol{P}\left( E + \text{Col}(G)\right) \mathbf{a} = ||\boldsymbol{P}\left( E\right) \mathbf{a} \rVert^2  + \left( \boldsymbol{P}(E^\perp)\mathbf{a}\right)^\top \boldsymbol{P}\left(\boldsymbol{P}(E^\perp )\boldsymbol{G} \right)\left( \boldsymbol{P}(E^\perp)\mathbf{a}\right) .$$

We now study the second term. Let $d = n-q$ and $\boldsymbol{Q}\in \R^{n\times d}$ be a matrix with columns being an orthonormal basis of $E^\perp$. Then, $ \boldsymbol{P}(E^\perp) = \boldsymbol{Q}\boldsymbol{Q}^\top$, and  $\boldsymbol{P}(E^\perp )\boldsymbol{G}  =\boldsymbol{Q}\boldsymbol{Q}^\top \boldsymbol{G}  := \boldsymbol{Q} \boldsymbol{H}$ where $\boldsymbol{H} = \boldsymbol{Q}^\top \boldsymbol{G}$ is standard Gaussian matrix since for each column $\mathbf{g}_j \sim \mathcal{N}(\boldsymbol{0}, \boldsymbol{I}_n)$, we have $\boldsymbol{Q}^\top  \mathbf{g}_j \sim \mathcal{N}(\boldsymbol{0}, \boldsymbol{Q}^\top \boldsymbol{Q}) =  \mathcal{N}(\boldsymbol{0}, \boldsymbol{I}_d).$ Note also that $$ \boldsymbol{P}\left(\boldsymbol{P}(E^\perp )\boldsymbol{G} \right) = \boldsymbol{P}\left(\boldsymbol{Q}\boldsymbol{H} \right)=\boldsymbol{Q}\boldsymbol{P}\left(\boldsymbol{H} \right)\boldsymbol{Q}^\top ,$$ since $\boldsymbol{Q}$ is orthonormal.  Therefore, we can write $$\left( \boldsymbol{P}(E^\perp)\mathbf{a}\right)^\top \boldsymbol{P}\left(\boldsymbol{P}(E^\perp )\boldsymbol{G} \right)\left( \boldsymbol{P}(E^\perp)\mathbf{a}\right)  = \mathbf{b}^\top \boldsymbol{P}(\boldsymbol{H})\mathbf{b},$$ with $\mathbf{b} := \boldsymbol{Q}^\top\mathbf{a}$ a deterministic vector. Since $\boldsymbol{H}$ is standard Gaussian, its column space is rotationally invariant in $\R^d$, which implies that the distribution of the quadratic form $\mathbf{b}^\top \boldsymbol{P}(\boldsymbol{H})\mathbf{b}$ depends on $\mathbf{b}$ only through its norm (see, e.g., the proof of Theorem 3.3.9. of \cite{vershynin2025high}). It follows that, choosing $\mathbf{g} \sim \mathcal{N}(\boldsymbol{0},\boldsymbol{I}_d)$ independent of $\boldsymbol{H}$, then $\mathbf{g} /||\mathbf{g}|| $ is a unit vector and thus,  $$\mathbf{b}^\top \boldsymbol{P}(\boldsymbol{H})\mathbf{b} \overset{d}{=} \rVert \mathbf{b}\rVert^2 \cdot \dfrac{\mathbf{g}^\top\boldsymbol{P}(\boldsymbol{H}) \mathbf{g}}{\mathbf{g}^\top\mathbf{g}}.$$

Note that $p<d$ and $\boldsymbol{H}\in \mathbb{R}^{d  \times p}$ is standard Gaussian, so it has rank $p$, and by the spectral theorem there exists an orthogonal matrix $\boldsymbol{R}$ such that 
$$\boldsymbol{P}(\boldsymbol{H})  = \boldsymbol{R}\begin{pmatrix}
\mathbf{I}_p & \mathbf{0} \\
\mathbf{0} & \mathbf{0}
\end{pmatrix}\boldsymbol{R}^\top.$$ Using rotational invariance again and the fact that $\mathbf{g}\indep \boldsymbol{H}$, we have $\boldsymbol{R}^\top \mathbf{g} \sim \mathcal{N}(\boldsymbol{0},\boldsymbol{I}_d)$, so that, $$\dfrac{\mathbf{g}^\top\boldsymbol{P}(\boldsymbol{H}) \mathbf{g}}{\mathbf{g}^\top\mathbf{g}} = \dfrac{\sum_{j=1}^p z_j^2}{\sum_{j=1}^p z_j^2 + \sum_{j=p+1}^d z_j^2}\sim \dfrac{\chi^2_{p}}{\chi^2_{p}+\chi^2_{d-p}}= \cB \left( \dfrac{p}{2},\dfrac{d-p}{2}\right),$$
where $z_1, ..., z_p$ are independent $\mathcal{N}(0,1)$. Therefore, noting that $d-p=n-q-p$, the result follows.

\end{proof}

\end{document}